%% file: main.tex
\documentclass[%
 reprint,
superscriptaddress,
 amsmath,amssymb,
pre,
]{revtex4-2}
\usepackage{styles/packages}
\usepackage{styles/commands}
\usepackage{graphicx}
\usepackage{dcolumn}
\usepackage{bm}
\usepackage{cancel}
\usepackage{todonotes}
\usepackage{thm-restate} 

\newcommand{\R}{\mathbb{R}}

\newcommand{\Z}{\mathbb{Z}}
\renewcommand{\C}{\mathbb{C}}

\newcommand{\F}{\mathcal{F}}

\begin{document}


\title{Ising Model Partition Function Computation as a \texorpdfstring{\\}{} Weighted Counting Problem}


\author{Shaan A. Nagy}
\affiliation{Department of Computer Science, Rice University, Houston, TX, USA}
\author{Roger Paredes}
\affiliation{Department of Civil and Environmental Engineering, Rice University, Houston, TX, USA}
\author{Jeffrey M. Dudek}
\affiliation{Department of Computer Science, Rice University, Houston, TX, USA}
\author{Leonardo Due\~nas-Osorio}
\affiliation{Department of Civil and Environmental Engineering, Rice University, Houston, TX, USA}
\author{Moshe Y. Vardi}
\affiliation{Department of Computer Science, Rice University, Houston, TX, USA}




\input{sections/0-Abstract}


\maketitle

\input{sections/1-Introduction}
\input{sections/2-Preliminaries}

\input{sections/3-Theory}
\input{sections/4-Experiments}

\input{sections/5-Conclusions}

\appendix
\input{sections/AppendixA-WCSPtoWMCReduction}
\input{sections/AppendixB-DichotomyThm}

\bibliography{main.bib}
\end{document}

%% file: sections/0-Abstract.tex
\begin{abstract}
While the Ising model remains essential to understand physical phenomena, its natural connection to combinatorial reasoning makes it also one of the best models to probe complex systems in science and engineering.  We bring a computational lens to the study of Ising models, where our computer-science perspective is two-fold: On the one hand, we consider the computational complexity of the Ising partition-function problem, or \#Ising, and relate it to the logic-based counting of constraint-satisfaction problems, or \#CSP. We show that known dichotomy results for \#CSP give an easy proof of the hardness of \#Ising and provide new intuition on where the difficulty of \#Ising comes from.
On the other hand, we also show that \#Ising can be reduced to Weighted Model Counting (WMC).  This enables us to take off-the-shelf model counters and apply them to \#Ising. We show that this WMC approach outperforms state-of-the-art specialized tools for \#Ising, thereby expanding the range of solvable problems in computational physics. 
\end{abstract}

%% file: sections/1-Introduction.tex
\section{Introduction}
The Ising spin glass model is a fundamental tool in statistical mechanics to study many-body physical systems \cite{barahona1982computational,liu_tropical_2021}. 
One important property of an Ising model is its \emph{partition function} (also called its \emph{normalization probability}) \cite{tanaka1980analytic}.
Spin-glass models have also been recently leveraged in data science, complex systems, and related problems, e.g. community detection \cite{eaton2012spin}.


Similarly, weighted counting is a fundamental problem in artificial intelligence, with applications in probabilistic reasoning, planning, inexact computing, and engineering reliability \cite{Bacchus2003,DH07,naveh2007constraint,CaiComplex,GSS08}.
The task is to count the total weight, subject to a given weight function, of the solutions of a given set of constraints \cite{GSS08}.
For example, weighted counting is a core task for many algorithms that normalize or sample from probabilistic distributions (arising from graphical models, conditional random fields, and skip-gram models, among others \cite{KF09}).

From a complexity-theoretic view, both weighted counting and the computation of the Ising-model partition function lie in the complexity class \#P-Hard \cite{Valiant79,barahona1982computational}.
Although proving the difficulty of \#P-Hard problems is still open, it is widely believed that \#P-Hard problems are fundamentally hard.
More specifically, under standard assumptions (e.g., weighted counting instances with rational, log-linear weights \cite{CFMV15,DFM20}) both weighted counting and the Ising model partition function are \#P-Complete.
This means that both problems are in some sense equivalent: computing an Ising model partition function can be done through a single weighted counting query with only a moderate (meaning polynomial) amount of additional processing, and vice versa.

In this work, we demonstrate that weighted counting provides useful lenses through which to view the problem of computing the partition function of an Ising spin glass model. We focus here on two closely-related formalizations of weighted counting: as a \emph{weighted constraint satisfaction problem (w\#CSP)} \cite{CaiComplex}, and as a \emph{weighted model counting (WMC) problem} \cite{GSS08}.
The study of w\#CSPs has largely been from a complexity-theoretic perspective, and has led to deep results on the computational complexity of counting (e.g., Theorem~\ref{restatthm:wcsp}) \cite{CaiComplex}. 
On the other hand, the study of WMC problems has largely focused on the development of practical tools (called \emph{counters}) for solving solving WMC instances. This has resulted in a huge variety of counters \cite{SBK05,OD15,LM17,DDV19,dudek2020parallel,DPV20,dudek2020dpmc,DFM20,FHWZ18,FHZ19}, which, despite the computational difficulty of counting, have been used to solve large, useful applied problems in a variety of fields (e.g., informatics, infrastructure reliability, etc.)~\cite{fichte2020model}. Section~\ref{sec:prelim} introduces Ising models and weighted counting.

Through one lens, viewing the partition function of an Ising model as a w\#CSP gives us foundational insights into its computational complexity. In Section \ref{sec:theory}, we highlight how a famous w\#CSP dichotomy theorem \cite{CaiComplex} can be used to classify Ising models according to the complexity of ground-state energy landscapes.

Through another lens, viewing the partition function of an Ising model as a WMC problem allows us to compute the partition function in practice.
We show in Section \ref{sec:wmc} that mature, powerful off-the-shelf WMC counters can be used to compute the partition function. In particular, we find that the counter \tool{TensorOrder} \cite{DDV19,dudek2020parallel} outperforms a variety of other counters and traditional approaches to computing the partition function of Ising models, including direct tensor network approaches in computational physics \cite{pan2020contracting}.

In contrast to approximation methods for evaluating the partition function~\cite{pan2020contracting}, our work focuses on exact methods. This strategy offers insights into algorithmic strategies to make partition function quantification accessible and accountable. However, as the partition function is parameterized by the inverse temperature parameter, $\beta\geq 0$, the case of $\beta \to \infty$, known as the zero-temperature limit, remains challenging for both exact and approximate tasks, including ground state ground-state energy and residual entropy quantification.

Nonetheless, Liu et al.~\cite{liu_tropical_2021} recently showed that replacing the usual sum and product binary operators for ordinary real numbers with the max and sum operators, respectively, known as the tropical algebra, gives access to partition function computations in the zero-temperature limit. We believe that the computational complexity analyses and Ising reductions to weighted model counters offered in our study for finite $\beta$ in partition function computations, pave the way for future tropical algebra counters, especially as they rely on the common language of tensor network contractions~\cite{dudek2019efficient,dudek2020parallel}.  

We conclude in Section \ref{sec:conc}, and offer ideas for future research and development.

\label{sec:intro}

%% file: sections/2-Preliminaries.tex
\section{Preliminaries}
\label{sec:prelim}
In this section, we give a versatile and formal definition of the \emph{Ising model}. In particular, we define the \emph{partition function} of an Ising model.

\subsection{Ising Models}
\label{sec:prelim:ising}

\begin{definition}
An \emph{Ising model} is a tuple $(\Lambda, J, h, \mu)$ where (using $\mathbb{R}$ to denote the real numbers):
\begin{itemize}
    \item $\Lambda$ is a set whose elements are called \emph{lattice sites},
    \item $J: \Lambda^2 \rightarrow \mathbb{R}$ is a function whose entry $J(i,j)$ is called the \emph{interaction} of $i, j \in \Lambda$,
    \item $h: \Lambda \rightarrow \mathbb{R}$ is a function, called the \emph{external field}, and
    \item $\mu \in \{-1, 1\}$.
\end{itemize}
A \emph{configuration} is a function $\sigma : \Lambda \rightarrow \{-1, 1\}$, which assigns a \emph{spin} (either $-1$ or $1$) to each lattice site. The (classical) \emph{Hamiltonian} is a function $H$ that assigns an \emph{energy} to each configuration, as follows:
\begin{equation}
    H(\sigma) \equiv - \sum_{i, j \in \Lambda} J(i,j) \sigma(i) \sigma(j) - \mu \sum_{j \in \Lambda} h(j) \sigma(j).
\end{equation}

\end{definition}

Notationally, $J$, $\sigma$, and $h$ accept arguments as subscripts and $\Lambda$ is implicit. Thus the Hamiltonian of an Ising model is written:
\begin{equation}
    H(\sigma) \equiv - \sum_{i, j} J_{ij} \sigma_i \sigma_j - \mu \sum_{j} h_j \sigma_j.
\end{equation}

An Ising model can be interpreted physically as a set of discrete magnetic moments, each with a ``spin'' of -1 or 1. The entries of $J$ indicate the strength of local interactions between two magnetic moments. The entries of $h$, together with $\mu$, indicate the strength and direction of an external magnetic field. It is also worth noting that the set of Ising Hamiltonians can be seen exactly as the set of degree-two polynomials over $\Lambda$ with no constant term.
A key feature of an Ising model is its set of \emph{ground states}, which are configurations that minimize the Hamiltonian (i.e., minimal-energy configurations). Finding a ground state of a given Ising model is famously NP-hard~\cite{barahona1982computational} (and a suitable decision variant is NP-complete). 
The connection between Ising models and Boolean satisfiability (SAT) has been especially important for the study of SAT phase transitions, where the behavior of randomized SAT problems can be understood through the behavior of \emph{spin-glass} models, which led to provable conjectures, c.f.~\cite{achlioptas2021number}.

An important problem for Ising models is the computation of the \emph{partition function}:
\begin{definition} \label{def:normalization}
The \emph{partition function} (or \emph{normalization probability}) of an Ising model $(\Lambda, J, h, \mu)$ with Hamiltonian $H$ at parameter $\beta \geq 0$ (called the \emph{inverse temperature}) is
\begin{equation}
    Z_\beta \equiv \sum_{\sigma : \Lambda \rightarrow \{-1, 1\}} e^{-\beta H(\sigma)}.
\end{equation}
\end{definition}
As its other name suggests, the partition function serves as a normalization constant when computing the probability of a configuration. These probabilities are given by a Boltzmann distribution: for a configuration $\sigma$, we have $P_\beta(\sigma) = \frac{e^{-\beta H(\sigma)}}{Z_\beta}$. 

There are several special cases of the Ising model that are studied in the literature. A few of the most common ones are:
\begin{itemize}
    \item \textbf{Ferromagnetic and Antiferromagnetic Ising Models} - The term ``Ising model'' is sometimes (especially historically) used to refer to the \emph{ferromagnetic Ising model}, which is the special case where all interactions are non-negative, i.e. where $J_{ij} \geq 0$ for all $i, j \in \Lambda$. Interactions where $J_{ij} < 0$ are called \emph{antiferromagnetic}.
    \item \textbf{2D Ising Models} - In a \emph{2D Ising model} the elements of $\Lambda$ are vertices of a 2D grid, and $J_{ij} = 0$ unless $i$ and $j$ are adjacent in the grid.
    \item \textbf{Sparse Interactions} - In many Ising models $J$ is sparse, i.e., most entries of $J$ are 0. 2D Ising models are an example of this. 
    \item \textbf{No External Field} - A common simplification is to assume that there is no external magnetic field, which means considering only Ising models where $h_i = 0$ for all $i \in \Lambda$.
    \item \textbf{No Self-Interaction} - In most studied cases $J_{ii} = 0$ for all $i \in \Lambda$. Interactions of a particle with itself only contribute a constant additive factor to the Hamiltonian.
    \item \textbf{Symmetric/Triangular Interactions} - We can always assume without loss of generality that, for all distinct $i, j \in \Lambda$, either $J_{ij} = 0$ or $J_{ji} = 0$. Similarly, we can instead assume without loss of generality that $J$ is symmetric. These assumptions are mutually exclusive unless $J$ is diagonal (which is uninteresting; see prior bullet).
\end{itemize}

Additionally, there is often interest in describing the physical connectivity of the lattices sites, as in 2D (planar) Ising models. This physical connectivity is sometimes called the \emph{topology} of the model, and information about a model's topology is useful in analyzing it. Such topologies may be represented as graphs, defined below.

\begin{definition} [(Simple) Graphs]
    A undirected graph (or, graph, for short)  is an pair $(V, E)$ where
    \begin{itemize}
        \item $V$ is a set whose elements are called \emph{vertices}.
        \item $E$ is a collection of sets of the form $\{v, w\}$, where $v$ and $w$ are distinct elements in $V$. Each set $\{v,w\}$ represents an \emph{edge} between vertices $v$ and $w$.
    \end{itemize}
\end{definition}
Typically, we restrict our attention to graphs that are \emph{finite} (i.e., $V$ is finite) and \emph{simple} (a vertex cannot have an edge with itself, and no two vertices can share more than one edge). 

\subsection{Weighted Counting Problems}
Often, computer scientists are interested in counting problems, problems involving counting the number of solutions in a solution space that satisfy certain constraints. For example, one might wish to know the number of satisfying solutions to a Boolean formula or the number of $3$-colorings for a particular graph. More than that, one might be interested in weighted counting problems, where one assigns weights to particular elements of a solution space and sums the weights over the solution space. Computation of the partition function of an Ising model is a good example of this case. We assign a weight ($e^{-\beta H(\sigma)}$) to each configuration of the model and sum these weights over all configurations to determine the partition function. This sort of problem is common across many fields and, in many cases, is computationally complex. 

\subsubsection{Weighted Counting Constraint Satisfaction Problems}
Many (but not all) weighted counting problems can be expressed as Weighted Counting Constraint Satisfaction Problems (w\#CSPs). Framing weighted counting problems as w\#CSPs introduces a standard form, 
which allows work done on the standard form to be applied to many different weighted counting problems. 
This is especially evident when analyzing the hardness of w\#CSPs which we discuss later (see Theorem~\ref{restatthm:wcsp}).

In the next definition we use $\mathbb{C}$ to denote the complex numbers and we use $\Z_+$ to denote the positive integers.
\begin{definition}[Constraint Language]
\label{def:constraint-language}
Let $D$ be a finite set and, for each $k \in \Z_+$, let $A_k$ be the set of all functions $F:D^k \rightarrow \mathbb{C}$.
Then a \emph{constraint language} $\F$ is a subset of $\bigcup\limits_{k \in \Z_+}A_k$.
\end{definition}

We write $\#CSP(\F)$ as the weighted \#CSP problem corresponding to $\F$.
\begin{restatable} [Weighted \#CSP - Instance and Instance Function] {restatdef}{wcspInstanceDef}
\label{restatdef:wcspInstanceDef}
Let $\F$ be a constraint language. Then an \emph{instance of $\#CSP(\F)$} is a pair $(I, n)$ where $n \in \Z_+$ and $I$ is a finite set of formulas each of the form $F(x_{i_1},\cdots,x_{i_k})$, where $F \in \F$ is a $k$-ary function and $i_1,\cdots,i_k \in [n]$ with each $x_{i_j}$ being a variable ranging over $D$. Note that $k$ and $i_1,\cdots,i_k$ may differ in each formula in $I$.\\
Given an instance $(I,n)$, we define the corresponding \emph{instance function} $F_I: D^n \rightarrow \C$ to be the following conjunction over $I$: 
\begin{equation}
    F_I(y) = \prod\limits_{F(x_{i_1},\cdots,x_{i_k}) \in I} F(y[i_1],\cdots,y[i_k]).
\end{equation}
The output of the instance $(I,n)$ is $Z(I) = \sum\limits_{y \in D^n} F_I(y)$.
\end{restatable}



As we will see, representing problems as w\#CSPs allows us access to a rich body of theoretical and computational work. In Section 3, we will see that converting a problem to a w\#CSP allows us to easily assess its computational complexity - how the number of operations needed to solve the problem scales with input size. Converting a problem to a w\#CSP also allows access to a slew of standardized computational tools. A common approach for computing the partition function of an Ising model is to represent the problem as a tensor network. Indeed, all w\#CSPs can be converted to tensor networks, so the tensor network approach can be motivated from this direction as well. 

\subsubsection{Weighted Model Counting}
\label{sec:prelim:wmc}

All weighted counting problems (and their w\#CSPs forms) can be reduced to weighted model counting (WMC). By converting a w\#CSP to WMC, we can take advantage of well-developed existing solvers which, in many cases, outperform problem-specific methods. Propositional model counting or the Sharp Satisfiability problem (\#SAT) consists of counting the number of satisfying assignments of a given Boolean formula. Without loss of generality, we focus on formulas in conjunctive normal form (CNF).
\begin{definition}[Conjunctive Normal Form (CNF)]\label{def:cnf}
A formula $\phi$ over a set $X$ of Boolean variables is in CNF when written as
\begin{equation}
	\phi = \bigwedge_{i=1}^{m} C_i =  \bigwedge_{i=1}^m \left( \bigvee_{j=1}^{k_i} l_{ij}\right),
\end{equation}
where every clause $C_i$ is a disjunction of $k_i\leq \vert X \vert$ literals, and every literal $l_{ij}$ is a variable in $X$ or its negation.
\end{definition}
Given a truth-value assignment $\tau:X\to \{0,1\}$ (such as a micro-state of a physical system), we use $\phi(\tau)$ to denote the formula that results upon replacing the variables $x\in X$ of $\phi$ by their respective truth-values $\tau(x)$, and say $\tau$ is a satisfying assignment of $\phi$ when $\phi(\tau)=1$. Thus, given an instance $\phi$ of \#SAT, one is interested in the quantity $\sum_{\tau\in[X]}\phi(\tau)$, where $[X]$ denotes the set of truth-value assignments.

Literal weighted model counting (WMC) is a generalization of {\#SAT} in which every truth-value assignment is associated to a real weight. Formally, WMC is defined as follows:

\begin{definition}[Weighted Model Counting]\label{def:wmc}
    Let $\phi$ be a formula over a set $X$ of Boolean variables, and let $W\colon X \times \{0,1\} \rightarrow \mathbb{R}$ be a function (called the weight function). Let $[X]$ denote the set of truth-value assignments $\tau\colon X \to \{0,1\}$. The weighted model count of $\phi$ w.r.t. $W$ is
    \begin{equation}
        W(\phi) \equiv \sum_{\tau \in [X]}\phi(\tau) \cdot \prod_{x \in X}W(x,\tau(x)).
    \end{equation}
    
\end{definition}
One advantage of WMC with respect to \#SAT is that it captures problems of practical interest such as probabilistic inference~\cite{sang2005performing} more naturally. Network reliability~\cite{duenas2017counting} and Bayesian inference~\cite{chavira2008probabilistic} are specific instances of probabilistic inference problems. Thus, the development of practical WMC solvers remains an active area of research, where algorithmic advances can enable the solution of difficult combinatorial problems across various fields.

In this paper we cast the Ising model partition function computation as a well understood problem of weighted model counting, which, in turn, enables its computation via actively developed off-the-shelf WMC solvers. Significantly, this paper gives empirical evidence that the runtime of \emph{exact} model counters vastly outperforms \emph{approximate} state-of-the-art physics-based tools that are currently used for partition function computations~\cite{agrawal2021partition}.

%% file: sections/3-Theory.tex
\section{Hardness and Relationship to Weighted Constraint Satisfaction}
\label{sec:theory}

It is well-known that computing the partition functions of Ising models is likely to be computationally intractcable, or as we explain below, \#P-hard~\cite{IsingSingleHardness}. In this section, we demonstrate how this can be easily derived from a formulation of partition-function computation as a w\#CSP. While the hardness of computing Ising-model partition functions is well-known, most proofs are extremely delicate and complex. The method we discuss below can be directly applied to many similar problems without the difficulties of a more classical proof (e.g., selecting an appropriate \#P-hard problem to reduce from, and constructing a clever reduction). Furthermore, we fold the discussion of hardness of computing the Ising-model partition functions into a much broader discussion of evaluating w\#CSPs, and from this we gain a better understanding of why partition functions are so hard to compute.

\subsection{Introduction to Computational Complexity}
\label{sec:theory:intro}

In this paper, we give a high-level overview of computational-complexity theory -- complexity theory, for short -- for counting problems. Readers interested in a more rigorous treatment of complexity theory might refer to Sipser's \textit{Introduction to the Theory of Computation}~\cite{sipserTOCIntro}.

One use of a general problem format such as w\#CSP is to be able to assess the computational complexity of a given problem. A problem is a mapping from problem instances to their associated outputs. When discussing problems, it is always important to be explicit about the set of instances over which a problem is defined (often called a problem's ``domain"). Hardness of a problem, as we see below, is a property of a problem's domain. It does not make sense to talk about the hardness of a problem instance, only the hardness of a problem, which consists of an infinite class of instances.

A problem's computational (time) complexity of a problem describes how the time (more formally the number of operations) required to solve a problem grows with the size of the input. 
Algorithms with low computational complexity are called \emph{scalable}, while algorithms with high computational complexity are not considered to be scalable.
Thus, there is often interest in knowing whether an algorithm with low computational complexity exists for a given problem~\cite{CLRS}.

A theoretical distinction is often made between problems that can be solved in polynomial time and those that are not believed to be solvable in polynomial time. There are many counting problems whose solvability in polynomial time is an open problem. The most famous example of such a problem is propositional model counting (\#SAT), which counts the number of assignments that satisfy a given formula in Boolean logic. Other examples include computing the number of maximum size cut sets of a graph (\#MAXCUT), computing matrix permanents for Boolean matrices and counting the number of perfect matchings of a bipartite graph~\cite{Valiant79}. In the study of counting problems, we often discuss problems in the complexity class $FP$ (Function Polynomial-Time) and problems which are $\#P$-hard. Problems in $FP$ are known to be solvable in polynomial time, and (for our purposes) problems which are $\#P$-hard are believed (but not known) to be impossible to solve in polynomial time~\cite{Valiant79}.

A useful tool in establishing membership of a problem in $FP$ or $\#P$-hard is the idea of polynomial equivalence. Problems that can be converted to one another in polynomial time are said to be polynomially equivalent. If two problems are polynomially equivalent, a polynomial-time algorithm for one problem can be used to construct a polynomial-time algorithm for the other problem. If a problem is polynomially equivalent to another problem in $FP$, both problems are in $FP$. Similarly, if a problem is polynomially equivalent to another problem that is $\#P$-hard, both problems are $\#P$-hard.


\subsection{The w\#CSP Dichotomy Theorem}
\label{sec:theory:dichot}
A landmark result in complexity theory is the following Dichotomy Theorem, which gives criteria for determining the complexity class of a w\#CSP~\cite{CaiComplex}. Every w\#CSP is either in $FP$, meaning that it can be solved in polynomial time, or $\#P$-hard, which for our purposes means it is not believed to be polynomially solvable. While the existence of this dichotomy is itself theoretically significant in that it helps us understand how computational hardness arises, this paper is interested in applying the theorem to determine the hardness of computing partition functions. Note also that the Dichotomy Theorem applies only to problems with finite constraint languages (Definition~\ref{def:constraint-language}). The conditions referenced in the following theorem are discussed below.

\begin{restatable} [Dichotomy Theorem for w\#CSPs~\cite{CaiComplex}] {restatthm}{wcsp}
\label{restatthm:wcsp}
Let $\F$ be a finite constraint language with algebraic complex weights. Then $\#CSP(\F)$ is polynomially solvable if the Block-Orthogonality Condition (Def \ref{condition:block-orth}), Mal'tsev Condition (Def \ref{condition:maltsev}), and Type Partition Condition hold (Def \ref{condition:type-partition}). Otherwise, $\#CSP(\F)$ is $\#P$-hard.
\end{restatable}

The statement of the w\#CSP Dichotomy Theorem and its conditions is rather complex, so it is useful to include some motivation. The three criteria given by the theorem follow naturally when we attempt to construct a polynomial-time algorithm for w\#CSPs. One reason that w\#CSPs are often computationally expensive to solve is that the effect of a single assignment to a given variable is hard to capture. In particular, given a problem instance $I$ and its associated n-ary formula $F_I(\boldsymbol{x})$, it is challenging to find a general rule that relates $F_I(x_1,x_2,\cdots,x_n)$ to $F_I(x'_1,x_2,\cdots,x_n)$. For this reason, current computational approaches require consideration of a set of assignments whose size is exponential in $n$. The algorithm driving the w\#CSP Dichotomy Theorem demands that changing the assignment of a particular variable changes the value of $F_I$ in a predictable way. The criteria of the Dichotomy Theorem are necessary and sufficient conditions for these assumptions to hold.

We now establish some definitions to help us describe the effect of an assignment on the value of $F_I$.
Recall the definition of an instance of a w\#CSP (Definition~\ref{restatdef:wcspInstanceDef}). An instance of $\#CSP(\F)$ is written as $(I,n)$ with $n$ a positive integer and $I$ a finite collection of formulas each of the form $F(x_{i_1},\cdots,x_{i_k})$. Here each $F \in \F$ is a $k$-ary function, $i_1,\cdots,i_k \in [n]$, and each variable $x_{i_j}$ ranges over $D$. Note that $k$ and $i_1,\cdots,i_k$ may differ in each formula in $I$. 
Given an instance $(I,n)$, we define $F_I: D^n \rightarrow \C$ as follows, where  $\boldsymbol{y}\in D^n$:
\begin{equation}
    F_I(\boldsymbol{y}) = \prod\limits_{F(x_{i_1},\cdots,x_{i_k}) \in I} F(y[i_1],\cdots,y[i_k]).
\end{equation}
So $F_I$ is a mapping that takes an assignment $\boldsymbol{y}$ to the variables ($\boldsymbol{x}$) and returns a value in $\C$, where the value assigned to each $x_{i_j}$ is given by $y[i_j]$. The output of the instance $(I,n)$ is the sum over all assignments $Z(I) = \sum\limits_{\boldsymbol{y} \in D^n} F_I(\boldsymbol{y})$, as in Definition~\ref{restatdef:wcspInstanceDef}. Note the similarity between this sum of products and our Ising partition function (Definition~\ref{def:normalization}).

Our first step is to break up the computation of $Z(I)$ so that we can understand the effect each variable assignment has on its value. We would like to be able to split $Z(I)$ into a sum of smaller sums, each restricted to a particular partial assignment. Then we can consider how differences between these partial assignments affect their contributions to $Z(I)$. 

In the development below, we take $t$ to be an arbitrary member of $\{1, \cdots, n\}$.
\begin{definition} [$F_I^{[t]}$]
Let $F_I^{[t]}: D^t \rightarrow \C$ be defined as 
\begin{equation}
    F_I^{[t]}(y_1, \cdots, y_t) = \sum\limits_{y_{t+1}, \cdots, y_n \in D} F_I(y_1, \cdots, y_n).
\end{equation}
\end{definition}
We can decompose $Z(I)$ into a sum of such terms. Thus $Z(I) = \sum\limits_{a \in D} F_I^{[1]}(a)$, and in general, we have that $Z(I) = \sum\limits_{\boldsymbol{y} \in D^t} F_I^{[t]}(\boldsymbol{y})$. 

We now consider $F_I^{[t]}$ as a $\abs{D}^{t-1} \times \abs{D}$ matrix, where $F_I^{[t]}(\boldsymbol{y},d) = F_I^{[t]}(y_1,\cdots,y_{t-1},d)$ for $\boldsymbol{y} \in D^{t-1}$ and $d\in D$:

\begin{align}
F_I^{[t]} &= \begin{bsmallmatrix}
F_I^{[t]}(\boldsymbol{y}^1, d_1) && \cdots && F_I^{[t]}(\boldsymbol{y}^1, d_{\abs{D}})\\
\vdots && \vdots && \vdots \\
F_I^{[t]}(\boldsymbol{y}^{\abs{D}^{t-1}}, d_1) && \cdots && F_I^{[t]}(\boldsymbol{y}^{\abs{D}^{t-1}}, d_{\abs{D}})\\
\end{bsmallmatrix}\\
&= \begin{bsmallmatrix}
\sum\limits_{\boldsymbol{w} \in D^{n-t}} F_I(\boldsymbol{y}^1, d_1, \boldsymbol{w}) && \cdots && \sum\limits_{\boldsymbol{w} \in D^{n-t}} F_I(\boldsymbol{y}^1, d_{\abs{D}}, \boldsymbol{w})\\
\vdots && \vdots && \vdots \\
\sum\limits_{\boldsymbol{w} \in D^{n-t}} F_I(\boldsymbol{y}^{\abs{D}^{t-1}}, d_1, \boldsymbol{w}) && \cdots && \sum\limits_{\boldsymbol{w} \in D^{n-t}} F_I(\boldsymbol{y}^{\abs{D}^{t-1}}, d_{\abs{D}}, \boldsymbol{w})\\
\end{bsmallmatrix}.
\end{align}

Writing $F_I^{[t]}$  as a matrix invites us to consider the following question: How does changing the assignment of a single variable change the contribution of a partial assignment to the instance's output $Z(I)$? In our notation, how does changing $d$ affect the value of $F_I^{[t]}(\boldsymbol{y}, d)$ for a given $\boldsymbol{y}$?
\begin{definition} [$F_I^{[t]}(\boldsymbol{y},\boldsymbol{\cdot})$]
We refer to the rows of the matrix $F_I^{[t]}$ by $F_I^{[t]}(\boldsymbol{y},\boldsymbol{\cdot})$ as below, with $\boldsymbol{y} \in D^{t-1}$ and $D = \{d_1, \cdots, d_{\abs{D}}\}$. Thus
\begin{equation}
F_I^{[t]}(\boldsymbol{y},\boldsymbol{\cdot}) = \begin{bsmallmatrix}
\sum\limits_{\boldsymbol{w} \in D^{n-t}} F_I(\boldsymbol{y}, d_1, \boldsymbol{w}), && \cdots, && \sum\limits_{\boldsymbol{w} \in D^{n-t}} F_I(\boldsymbol{y}, d_{\abs{D}}, \boldsymbol{w})\\
\end{bsmallmatrix}.
\end{equation}
\end{definition}

We now have a way to relate similar partial assignments. In particular, for each partial assignment $\boldsymbol{y}$, the vector $F_I^{[t]}(\boldsymbol{y},\boldsymbol{\cdot})$ captures the contributions to $Z(I)$ of partial assignments to the first $t$ variables that agree with $\boldsymbol{y}$ on the first $t-1$ variables. These vectors 
are now our objects of interest. 

The set of such vectors $F_I^{[t]}(\boldsymbol{y},\boldsymbol{\cdot})$ is quite large, since there are $\abs{D}^{t-1}$ choices for $\boldsymbol{y}$. We would like to reduce the computation and information needed to manage this set. We make the observation that if two vectors $F_I^{[t]}(\boldsymbol{y},\boldsymbol{\cdot})$ and $F_I^{[t]}(\boldsymbol{y'},\boldsymbol{\cdot})$ are scalar multiples of one another, then if we know $F_I^{[t]}(\boldsymbol{y},\boldsymbol{\cdot})$ and a non-zero entry of $F_I^{[t]}(\boldsymbol{y'},\boldsymbol{\cdot})$, we can compute $F_I^{[t]}(\boldsymbol{y'},\boldsymbol{\cdot})$ easily. This motivates us to introduce an equivalence relation on our set of partial assignments $\boldsymbol{y}$. Given $\boldsymbol{y}, \boldsymbol{y'} \in D^{t-1}$, we say that $\boldsymbol{y} \equiv \boldsymbol{y'}$ if 
the vectors
$F_I^{[t]}(\boldsymbol{y},\boldsymbol{\cdot})$ and $F_I^{[t]}(\boldsymbol{y'},\boldsymbol{\cdot})$ are scalar multiples of one another. This notion is formalized below in Definition~\ref{def:partial-assignment-partition}.

If we have a small number of equivalence classes and if it is easy to determine to which equivalence class a partial assignment $\boldsymbol{y}\in D^{t-1}$
belongs to, then using this equivalence relation would make it easier to compute and manage the $F_I^{[t]}(\boldsymbol{y},\boldsymbol{\cdot})$ vectors. We see below that the dichotomy criteria of Theorem~\ref{restatthm:wcsp} guarantee both these desiderata.



If we want our partition to be useful, we need to assign to each equivalence class a representative element. Then for each $\boldsymbol{y}$ in the equivalence class, we can compute
$F_I^{[t]}(\boldsymbol{y},\boldsymbol{\cdot})$  by scaling this representative. 
We construct such a vector of dimension $\abs{D}$ as follows.
\begin{definition} [$\boldsymbol{v}^{\boldsymbol{y}}$] 
\label{def:partial-assignment-representative}
Given $\boldsymbol{y} \in D^{t-1}$ and a total order on $D$, we define
\begin{equation}
    \boldsymbol{v}^{\boldsymbol{y}} = \frac{F_I^{[t]}(\boldsymbol{y}, \boldsymbol{\cdot})}{F_I^{[t]}(\boldsymbol{y}, a_t)}    
\end{equation}
for the least $a_t \in D$ so that $F_I^{[t]}(\boldsymbol{y}, a_t) \neq 0$. We say $\boldsymbol{v}^{\boldsymbol{y}} = 0$ if no such $a_t$ exists. \end{definition}

Observe that for an arbitrary $\boldsymbol{y} \in D^{t-1}$, we have that $F_I^{[t]}(\boldsymbol{y},\boldsymbol{\cdot}) = F_I^{[t]}(\boldsymbol{y}, a_t)\cdot\boldsymbol{v}^{\boldsymbol{y}}$. In addition, given $\boldsymbol{y}, \boldsymbol{y'} \in D^{t-1}$, if $y \equiv y'$ as discussed earlier, then $\boldsymbol{v}^{\boldsymbol{y}} = \boldsymbol{v}^{\boldsymbol{y'}}$. This motivates the following formalization of the equivalence relation discussed above.

\begin{definition} [$S_{[t,j]}$]
\label{def:partial-assignment-partition}
Given $\boldsymbol{y}, \boldsymbol{y'} \in D^{t-1}$, we say that $\boldsymbol{y}$ and $\boldsymbol{y'}$ are equivalent,
denoted $\boldsymbol{y} \equiv_t \boldsymbol{y'}$ if 
$\boldsymbol{v}^{\boldsymbol{y}} = \boldsymbol{v}^{\boldsymbol{y'}}$. We denote the equivalence classes induced by this equivalence relation $S_{[t,1]}, \cdots, S_{[t,m_t]}$. 
\end{definition}
We often use $\equiv$ instead of $\equiv_t$, when $t$ is clear from the context.

\begin{figure*}[hbt!]
	\centering
    \includegraphics[scale=0.50]{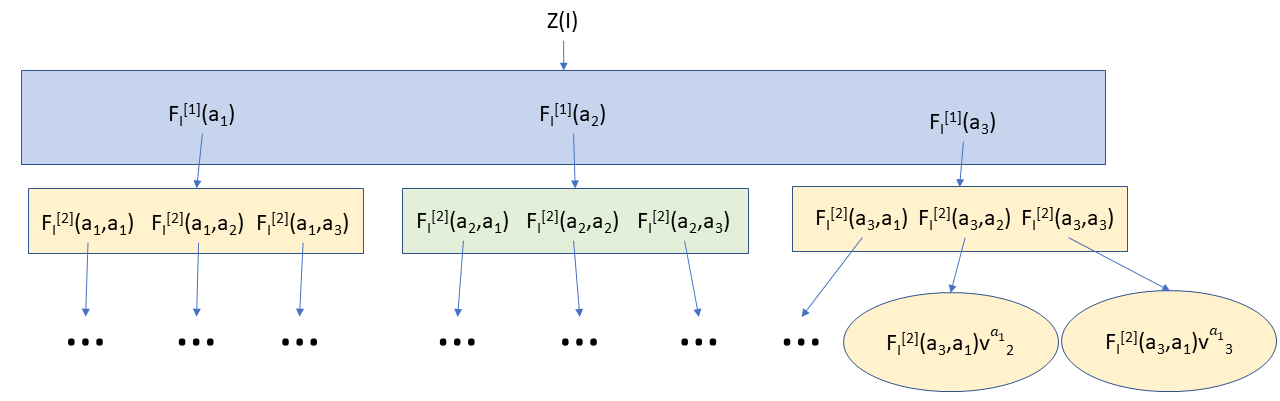}
    \caption{Diagram of contributions of partial assignments to $Z(I)$. The colors of each vector denote its equivalence class, thus $a_1 \equiv a_3$. When computing $F_I^{[2]}(a_3, a_2)$, we can use $\boldsymbol{v}^{a_1}$, which we already know from $F_I^{[2]}(a_1, \cdot)$. In this way, we need only check one partial assignment to 2 variables ($a_3, a_1$) to determine $F_I^{[2]}(a_3)$ when we otherwise would have needed to check three. As we work further down the tree, repeated applications of this shortcut result in an exponential speedup (assuming we have sufficiently few equivalence classes, guaranteed by the criteria of Theorem~\ref{restatthm:wcsp}).}
	\label{fig:CSPDiagram}
\end{figure*}
We now have a way to discuss symmetries between partial assignments in $D^{t-1}$. We can consider computing $Z(I)$ in this framework to understand what advantages we have gained. This computation can be visualized as in Figure~\ref{fig:CSPDiagram}. We regard $Z(I)$ as $F_I^{[0]}()$ and compute it recursively, traversing the tree in Figure~\ref{fig:CSPDiagram} in a depth-first traversal.

Consider the vector $\boldsymbol{v} = \frac{F_I^{[1]}(\boldsymbol{\cdot})}{F_I^{[1]}(a_1)} \in D^{\abs{D}}$, where $a_1$ is the least element of $D$ such that $F_I^{[1]}(a_1) \neq 0$ as in Definition~\ref{def:partial-assignment-representative}. We denote entries of $\boldsymbol{v}$ as $\boldsymbol{v}[a] = \frac{F_I^{[1]}(a)}{F_I^{[1]}(a_1)}$. Observe
\begin{align}
  Z(I) = F_I^{[0]}() &= \sum\limits_{a \in D} F_I^{[1]}(a)\\
       &= \sum\limits_{a \in D} F_I^{[1]}(a_1) \frac{F_I^{[1]}(a)}{F_I^{[1]}(a_1)} \\
       &= \sum\limits_{a \in D} F_I^{[1]}(a_1) \boldsymbol{v}[a] \\
       &= F_I^{[1]}(a_1) \sum\limits_{a \in D} \boldsymbol{v}[a].
\end{align}

So to compute $Z(I)$, one needs to compute $F_I^{[1]}(a_1)$ and $\boldsymbol{v}$.
We first compute $F_I^{[1]}(a_1)$. As per Definition~\ref{def:partial-assignment-representative}, consider the least $a_2 \in D$ 
so that $F_I^{[2]}(a_1, a_2) \neq 0$ and write $\boldsymbol{v}^{a_1} = \frac{F_I^{[2]}(a_1, \boldsymbol{\cdot})}{F_I^{[2]}(a_1,a_2)}$. Then 
\begin{align}
F_I^{[1]}(a_1) &= \sum\limits_{b \in D} F_I^{[2]}(a_1, b)\\
               &= \sum\limits_{b \in D} F_I^{[2]}(a_1, a_2) \frac{F_I^{[2]}(a_1,b)}{F_I^{[2]}(a_1,a_2)} \\
               &= \sum\limits_{b \in D} F_I^{[2]}(a_1, a_2) \left(\boldsymbol{v}^{a_1}[b]\right) \\
               &= F_I^{[2]}(a_1, a_2) \sum\limits_{b \in D} \boldsymbol{v}^{a_1}[b].
\end{align}
In order to finish computing $F_I^{[1]}(a_1)$, we must compute $F_I^{[2]}(a_1, a_2)$ and $\boldsymbol{v}^{a_1}$. We continue walking down our tree in this way until we reach the leaves where $F_I^{[n]} = F_I$ is easily computable. When we compute our $\boldsymbol{v}^{\boldsymbol{y}}$ vectors, we find each entry as per Definition~\ref{def:partial-assignment-representative} by visiting unexplored branches of the tree in a depth-first fashion. However, we can apply the  equivalence
relation $\equiv$ to avoid computing $\boldsymbol{v}^{\boldsymbol{y}}$ when it is known from earlier computation of an equivalent partial assignment. This speedup is explained next in the context of the top level of our tree.

After computing $F_I^{[1]}(a_1)$, we must compute $\boldsymbol{v}$. For each $a_1' \in D$ such that $a_1 \equiv a_1'$, we have that 
\begin{align}
\boldsymbol{v}[a_1'] &= \frac{F_I^{[1]}(a_1')}{F_I^{[1]}(a_1)}\\
&= \frac{1}{F_I^{[1]}(a_1)} F_I^{[2]}(a_1', a_2) \sum\limits_{b \in D} (\boldsymbol{v}^{a_1'}[b])\\
&= \frac{1}{F_I^{[1]}(a_1)} F_I^{[2]}(a_1', a_2) \sum\limits_{b \in D} (\boldsymbol{v}^{a_1}[b]).
\end{align}

Since we have already computed $\boldsymbol{v}^{a_1}$ in our computation of $F_I^{[1]}(a_1)$, we need only compute $F_I^{[2]}(a_1', a_2)$ to determine $\boldsymbol{v}[a_1']$. This reduces the number of partial assignments of length 2 we must explore to determine $\boldsymbol{v}[a_1']$ by a factor of $n$ (see Figure~\ref{fig:CSPDiagram}). As we proceed with our computation of various $\boldsymbol{v}^{\boldsymbol{y}}$ and $F_I^{[t]}(\cdot)$, the ability to reuse previously computed $\boldsymbol{v}$'s substantially reduces the search space we must explore (supposing the number $m_t$ of equivalence classes of $\equiv_t$ is small enough [$m_t << \abs{D}^{t-1}$ for each $t \in \{1,\cdots, n\}$]).

More generally, we may write
\begin{equation}
    Z(I) = F_I(a_1,\cdots,a_n) \prod\limits_{t \in \{1, \cdots, n\}}\left(\sum\limits_{b \in D}\boldsymbol{v}^{a_1, \cdots, a_{t-1}}[b]\right)
\end{equation}
and as we have just seen, the computation of each $\boldsymbol{v}^{a_1, \cdots, a_{t-1}}$ is made easier each time we have $(a_1, \cdots, a_{t-1}, d) \equiv (a_1, \cdots, a_{t-1}, d')$ for distinct $d, d' \in D$.\\

For this approach to be useful, we must be able to determine in polynomial time to which equivalence class $S_{[t,j]}$ a given $\boldsymbol{y}$ belongs (so that we can  reuse $\boldsymbol{v}^{\boldsymbol{y}}$). The Dichotomy Theorem criteria are precisely the conditions required for the collection of these data to be determined in polynomial time.

While the complete formulation of the Dichotomy Theorem for w\#CSPs is for complex-valued constraints, we simplify the statements here to cover only the real case. Interested readers can refer to~\cite{CaiComplex} for coverage of the general complex case or to~\cite{CaiNonNegative} for the simpler non-negative case. First, we go over the three criteria of the Dichotomy Theorem.

The \emph{Block-Orthogonality Condition} guarantees that the number $m_t$  of equivalence classes of $\equiv_t$ is not too large.
\begin{definition} [Block-Orthogonal]
 Consider two vectors $\boldsymbol{a}$ and $\boldsymbol{b} \in \R^k$, and define $\boldsymbol{\abs{a}} = (\abs{a_1}, \cdots, \abs{a_k})$ and $\boldsymbol{\abs{b}} = (\abs{b_1}, \cdots, \abs{b_k})$. Then $\boldsymbol{a}$ and $\boldsymbol{b}$ are said to be block orthogonal if the following hold:
 \begin{itemize}
     \item $\boldsymbol{\abs{a}}$ and $\boldsymbol{\abs{b}}$ are linearly dependent (i.e. they are scalar multiples of one another).
     \item For every distinct value $a \in \{\abs{a_1}, \cdots, \abs{a_k}\}$, letting $T_a = \{j \in [k] : \abs{a}_j = a\}$ be the set of all indices $j$ on which $\abs{a}_j = a$, we have $\sum\limits_{j \in T_a}a_j b_j = 0$.
 \end{itemize}
\end{definition}

Note that if two vectors $\boldsymbol{a}$ and $\boldsymbol{b} \in \R^k$ are block orthogonal, then $\abs{\boldsymbol{a}}$ and $\abs{\boldsymbol{b}}$ are linearly dependent, but $\boldsymbol{a}$ and $\boldsymbol{b}$ need not be.

\begin{definition} [Block-Orthogonality Condition] \label{condition:block-orth} We say that a constraint language $\F$ satisfies the block orthogonality condition if, for every function $F \in \F$, $t \in [n]$, and $\boldsymbol{y},\boldsymbol{z} \in D^{t-1}$, the row vectors $F^{[t]}(\boldsymbol{y},\boldsymbol{\cdot})$ and $F^{[t]}(\boldsymbol{z},\boldsymbol{\cdot})$ are either block orthogonal or linearly dependent.


    
\end{definition}
 

The Dichotomy Theorem for w\#CSPs (Theorem~\ref{restatthm:wcsp}) requires two other criteria: the Mal'tsev Condition and the Type-Partition Condition. These conditions make computing membership in $S_{[t,j]}$ tractable. They are reviewed in Appendix~\ref{appendix:dichot} and more thoroughly in~\cite{CaiComplex}, but they are unnecessary for the remainder of the paper.


When possible, representing an arbitrary weighted counting problem (such as Ising partition function computation) as a w\#CSP  allows us to more easily assess its hardness, and indeed this analysis can produce a polynomial-time algorithm to solve the problem when it is found to be in \textit{FP}. It should be noted that the question of determining whether or not a problem with complex (or even real-valued) weights meets or fails these dichotomy criteria is not known to be decidable~\cite{CaiComplex}. While the Block-Orthogonality Condition is easy to check, the Mal'tsev and Type-Partition Conditions, which impose requirements on all $F_I$'s, are not.
In the non-negative real case, there is a separate but analogous dichotomy theorem, which is known to be decidable~\cite{CaiNonNegative}. Regardless, in many cases it is not hard to determine manually whether a given problem instance satisfies the criteria. An example in the case of the Ising problem can be found in Lemma~\ref{lemma:FaNotBlock}.

\subsection{Ising as a w\#CSP}
For Ising models, it is not hard to reduce computation of the partition function to a w\#CSP. In doing so, however, we must be explicit about the instances over which such a formulation is defined. The Ising problem corresponding to $\#CSP(\text{Ising})$ includes all instances whose 
topologies are (finite) simple graphs at any temperature. 
Note that there is no requirement that $h = 0$. 

\begin{definition} [$\#CSP(\text{Ising})$]
Let $A_1$ be the set of functions $\{f: \{-1,1\} \rightarrow \R \mid f(a) = \lambda^{a} \mid \lambda > 0\}$, and let $A_2$ be the set of functions $\{f: \{-1,1\}^2 \rightarrow \R \mid f(a,b) = \lambda^{ab} \mid \lambda > 0\}$. Take $\F = A_1 \cup A_2$. Then define $\#CSP(\text{Ising}) = \#CSP(\F)$.
\end{definition}
The functions in $A_1$ correspond to the effect of 
the external field ($h$) on the lattice sites, 
and the functions in $A_2$ correspond to interactions between lattice sites. 

Often we restrict our constraint language to only include rational-valued functions. This makes all the values we deal with finitely representable, which is required for our problem to be computable. For simplicity, we omit this detail from our definition of $\#CSP(\text{Ising})$. Appropriately modified versions of all subsequent claims still hold when we make this restriction (esp. Lemma~\ref{lemma:polyEquiv} and Theorem~\ref{thm:IsingHard}). In real-world computation, values are expressed as floating points~\cite{FPArithmetic}, and we accept the small inaccuracies that entails. Thus, the instances we are interested in in practice are rational-valued. Readers interested in attempts to eliminate floating point inaccuracies might refer to~\cite{InfPrecisionArithmetic}.

Given the structural similarities between computation of an Ising partition function and $\#CSP(\text{Ising})$, the equivalence of these problems should not be surprising. Recall that polynomial equivalence is discussed at the end of Section~\ref{sec:theory:intro}.

\begin{lemma}
\label{lemma:polyEquiv}
$\#CSP(\text{Ising})$ is  polynomially equivalent to the following problem: Given an inverse temperature and an Ising model whose topology is a (finite) simple graph, compute the associated partition function. 
\end{lemma}

\begin{proof}
Given an Ising Model $(\Lambda, J, h, \mu)$ and an inverse temperature $\beta \geq 0$, we set
\begin{itemize}
    \item $F_{i,j} \in \F$ so that $F_{i,j}(1,1) = F_{i,j}(-1,-1) = e^{\beta J_{i,j}}$ and $F_{i,j}(1,-1) = F_{i,j}(-1,1) = e^{-\beta J_{i,j}}$.
    \item $F_i \in \F$ so that $F_i(1) = e^{\beta \mu h_i}$ and $F_i(-1) = e^{-\beta \mu h_i}$.
\end{itemize}

Then we have an instance $(I,n)$ of $\#CSP(\text{Ising})$ where $n = \abs{\Lambda}$ and $I = \{F_{i,j}(i,j): i,j \in \Lambda\} \cup \{h_i(i): i \in \Lambda\}$.
It is easy to see that $Z(I) = Z_\beta$; the solution to the instance of $\#CSP(\text{Ising})$ is equal to the partition function at the inverse temperature $\beta$.

It is also easy to see that any instance of $\#CSP(\text{Ising})$ can be converted (in polynomial time) to an instance of partition function computation of an Ising model by the reverse construction. \end{proof}Thus $\#CSP(\text{Ising})$ and computing partition functions of Ising models are polynomially equivalent problems. So, we can easily determine the hardness of the problem of computing the partition function of the Ising model by applying the Dichotomy Theorem to $\#CSP(\text{Ising})$. However, the Dichotomy Theorem for w\#CSP applies only to finite constraint languages. We have defined $\#CSP(\text{Ising})$ with an infinite constraint language, so we cannot directly apply the Dichotomy Theorem to $\#CSP(\text{Ising})$.

To discuss hardness in a meaningful way, we consider a simple finite subproblem of $\#CSP(\text{Ising})$ (i.e. a w\#CSP whose constraint language is a finite subset of that of $\#CSP(\text{Ising})$). If we show that this subproblem is \#P-hard, then $\#CSP(\text{Ising})$ must also be \#P-hard. Intuitively, if $\#CSP(\text{Ising})$ is believed to have no polynomial time solution on a subset of its valid instances, then there should be no solution for $\#CSP(\text{Ising})$ that is polynomial time on the entire set of valid problem instances. In general, if a subproblem of a problem is \#P-hard, then the problem itself is \#P-hard as well.

For example, one can consider weighted model counting of Boolean formulas in CNF (Definition~\ref{def:cnf}) as a subproblem of WMC (Definition~\ref{def:wmc}). Weighted model counting of CNF formulas can be shown to be \#P-hard, and from this fact we can determine that WMC is \#P-hard as well. Note that the converse does not hold; not all subproblems of a \#P-hard problem are guaranteed to be \#P-hard.

The instances of our subproblem of $\#CSP(\text{Ising})$ correspond to Ising models with no external field ($h=0$) whose interactions $J$ are integer multiples of some positive $a \neq 1$. We construct this subproblem as follows:
Let $a' \neq 1$ be a positive rational number, and define $a = \log{a'}$. Define $\F_a = \{f_a\}$, where $f_a: \{-1,1\}^2 \rightarrow \R$ is 
\begin{equation}
    f_a(x,y) = \begin{cases} e^a & x=y\\
e^{-a} & x \neq y
\end{cases}.
\end{equation}
Clearly, $\F_a$ is a subset of the constraint language used to define $\#CSP(\text{Ising})$, so $\#CSP(\F_a)$ is a subproblem of $\#CSP(\text{Ising})$. 

\begin{lemma}\label{lemma:FaNotBlock}
$\#CSP(\F_a)$ does not satisfy the Block-Orthogonality Condition
\end{lemma}
\begin{proof} We will show that $\#CSP(\F_a)$ does not satisfy the Block-Orthogonality Condition, and thus that it is \#P-hard. Consider $F^{[2]}$ as a matrix, that is:
\begin{equation}
F^{[2]} = \begin{bmatrix}e^a & e^{-a}\\ e^{-a} & e^{a}\end{bmatrix}.
\end{equation}
Clearly, the rows of $F^{[2]}$ are not linearly dependent ($\frac{e^a}{e^{-a}} = e^{2a} \neq e^{-2a} = \frac{e^{-a}}{e^{a}}$). Furthermore, $\abs{F^{[2]}(0)}$ and $\abs{F^{[2]}(1)}$ are not linearly dependent, so the rows of $F^{[2]}$ are not block-orthogonal. We conclude that $\#CSP(\F_a)$ does not satisfy the Block-Orthogonality Condition.\end{proof} Since $\#CSP(\F_a)$ does not satisfy the Block-Orthogonality Condition, it is \#P-hard  by our earlier Dichotomy Theorem. The next theorem follows.

\begin{theorem}
\label{thm:IsingHard}
$\#CSP(\text{Ising})$ is $\#P$-hard.
\end{theorem}

To summarize, we first saw that computing the partition function of Ising models whose interactions are integer multiples of some positive rational $a \neq 1$ is \#P-hard. We then determined that computing the partition function of Ising models in general is \#P-hard. This result followed immediately from a straightforward application of the Dichotomy Theorem for w\#CSPs rather than a more laborious bespoke reduction. Furthermore, our application of the Dichotomy Theorem provides some intuition on where the difficulty of Ising partition function computation comes from (i.e., the difficulty in relating the change of a lattice site's spin to a change in probability across configurations).

Insofar as hardness is concerned, we might also consider the case where we restrict Ising models' interactions to take on values of either $a$ or $0$. In many physical systems of interest, only a single type of interaction with a fixed strength can occur (e.g. in models of ferromagnetism with only nearest-neighbor interactions)~\cite{IsingOverview}. This formulation of the Ising problem is unfortunately not easily expressible as a w\#CSP, since constraint functions can be applied arbitrarily many times to the same inputs. While the w\#CSP approach is not applicable here, it is known that even in this restricted case computing the Ising partition function is \#P-hard. This follows from a reduction of \#MAXCUT to polynomially many instances of Ising model partition function computation~\cite{IsingSingleHardness}. This reduction from \#MAXCUT gives a stronger statement than our application of the w\#CSP Dichotomy Theorem, but it requires much more work than applying an out-of-the-box theorem. Specifically, the instances of Ising partition function computation are used to approximate the integer-valued solution to \#MAXCUT, and the solution to \#MAXCUT is recovered once the error of the approximate is sufficiently small ($<0.5$).

There is also often interest in hardness when we restrict our attention to planar Ising models~\cite{barahona1982computational}. This case too is not easily represented as a w\#CSP, since w\#CSPs do not give us control over problem topologies. However, it is known that the partition functions of planar Ising models are polynomially computable~\cite{barahona1982computational}. The typical approach is to reduce the problem of computing the partition function to a weighted perfect matching problem and apply the Fisher-Kasteleyn-Temperley (FKT) Algorithm. This solution is given in great detail in~\cite{barahona1982computational}. 

The inability of w\#CSP to capture these Ising subproblems demonstrates the limits of the w\#CSP framework.
There has been recent work on \emph{Holant Problems}, a class of problems broader than w\#CSP, which adds the ability to control how many times variables may be reused~\cite{CaiHolant}. The Holant framework still does not allow the degree of specificity needed to capture these Ising formulations, although it is perhaps a step in the right direction. While these problem classes and their accompanying dichotomy theorems are intended to answer deep theoretical questions about the source of computational hardness, applying them to real-world problems gives very immediate and useful hardness results. Application to real-world problems also suggests where the expressiveness of problem classes like w\#CSP can be improved. As the expressiveness of these problem classes increases, so does our ability to determine the hardness of precisely defined subproblems~\cite{CaiHolant}.

%% file: sections/4-Experiments.tex
\section{Weighted Model Counting}
\label{sec:wmc}

Every instance of a weighted counting constraint satisfaction problem (w\#CSP, Def.~\ref{restatdef:wcspInstanceDef}) can be reduced to an equivalent instance of weighted model counting (WMC, Def.~\ref{def:wmc}), where the domain of variables is restricted to the Boolean domain and the formulas are expressed in conjunctive normal form (CNF, Def.~\ref{def:cnf}). Appendix~\ref{appendix:wcsp2wmc} gives a general reduction from w\#CSP to WMC. This reduction to CNF WMC can be advantageous, as there are mature and dedicated WMC solvers that can handle very large CNF formulas~\cite{fichte2020model}. Therefore, a profitable alternative to direct computation of a weighted counting constraint-satisfaction
problem (e.g., Ising partition function computation) is to develop a succinct reduction from the w\#CSP problem of interest to WMC.

The remainder of this section reduces the problem of computing the partition function of the Ising model to one of weighted model counting. We then apply a diverse suite of WMC solvers and empirically evaluate their performance against state-of-the-art problem-specific tools.

\subsection{Ising model partition function as WMC}
We now reduce the problem of computing the partition function of the Ising model, $Z_{\beta}$~(Def.~\ref{def:normalization}), to one of weighted model counting~(Def.~\ref{def:wmc}). Formally, our reduction takes an Ising model and constructs an instance $W(\phi)$ of weighted model counting such that $Z_\beta=W(\phi)$. The reduction consists of two steps. In Step~1, we construct a CNF formula, $\phi$, over a set $X$ of Boolean variables, such that the truth-value assignments $\tau\in\{0,1\}^X$ that satisfy the formula correspond exactly to the configurations $\sigma\in\{\pm 1\}^\Lambda$ of the Ising model. In Step~2, we construct a literal weight function, $W$, such that for every satisfying assignment, $\tau$, the product of literal weights, $\prod_{x\in X} W(x,\tau(x))$, is equal to the Boltzmann weight of the corresponding Ising model configuration $\sigma$. Together, Step~1 and Step~2 ensure that $Z_\beta=W(\phi)$, which we make explicit at the end of this section. We next show the details of the two steps.

\subsubsection{Step 1 of 2: The Boolean formula in CNF}
To construct the Boolean formula $\phi$ in CNF, first introduce the set $X$ of Boolean variables $x_{i}$, with $i\in \Lambda$, and Boolean  variables $x_{ij}$, with $i,j \in \Lambda$ and $i\neq j$~\footnote{{Whenever $J_{ij}=J_{ji}=0$, we can ignore variable $x_{ij}$ for $i\neq j$}}. We use the variable $x_i$ to represent the spin state $\sigma_i\in\{\pm1\}$ for every $i\in\Lambda$. Also, we use the variable $x_{ij}$ to represent the interaction $\sigma_i\sigma_j\in\{\pm1\}$ for every $i\neq j$. By convention, we associate the truth-value one with the positive sign, and the truth-value zero with the negative sign. The semantic relation between truth-value assignments $\tau \in \{0,1\}^X$ and Ising model configuration $\sigma\in\{\pm 1\}^\Lambda$ is formally encoded by the below equations and Table~\ref{tab:predicate}:
\begin{equation}\label{eq:sigma2tau}
	\begin{aligned}
		&\sigma_i = 2\tau(x_i) - 1, & \mbox{for } i \in \Lambda,\\
		&\sigma_i \sigma_j = 2\tau(x_{ij}) - 1, \quad &\mbox{for } i,j \in \Lambda: i\neq j.
	\end{aligned}
\end{equation}
The first equation ensures that truth-value assignments $\tau(x_i)=1$ and $\tau(x_i)=0$ correspond exactly to spin states $\sigma_i=+1$ and $\sigma_i=-1$ respectively. The second equation ensures that truth-value assignments, $\tau(x_{ij})=1$ and $\tau(x_{ij})=0$, correspond exactly to positive and negative interactions, $\sigma_i = \sigma_j$ and $\sigma_i \neq \sigma_j$, respectively. Not every truth-value assignment $\tau\in \{0,1\}^X$ is consistent with the encoding of Eq.~(\ref{eq:sigma2tau}). Specifically, the set of truth-value assignments that satisfy the relations of Eq.~(\ref{eq:sigma2tau}) must be such that $\tau(x_{ij})=1$ exactly when $\tau(x_{i})=\tau(x_{j})$ (i.e., when the spins have equal sign $\sigma_i=\sigma_j$), and $\tau(x_{ij})=0$ otherwise. This is consistent with the signs of interactions, which are positive for equal spin states and negative otherwise. The set of assignments that is consistent with the encoding of Eq.~(\ref{eq:sigma2tau}) is exactly the set of satisfying assignments of the Boolean formula:
\begin{equation}\label{eq:phi}
	\phi \equiv \bigwedge_{i,j:i\neq j} p_{ij},
\end{equation}
with $p_{ij}=[(x_i \Leftrightarrow x_j) \Leftrightarrow x_{ij}]$. We can equivalently express $p_{ij}$ in CNF as the conjunction of the clauses $(x_{i} \vee \bar{x}_{j} \vee \bar{x}_{ij} )$, $(x_{i} \vee x_{j} \vee x_{ij} )$, $(\bar{x}_{i} \vee x_{j} \vee \bar{x}_{ij} )$, and $(\bar{x}_{i} \vee \bar{x}_{j} \vee x_{ij} )$. The correctness of this CNF encoding can be shown by first noting that a satisfying assignment of $\phi$ must satisfy all predicates $p_{ij}$. Then, from the truth-table of predicate $p_{ij}$ depicted in Table~\ref{tab:predicate}, it follows that Eq.~(\ref{eq:sigma2tau}) maps every satisfying assignment $\tau$ of $\phi$ to an Ising model configuration $\sigma$, and vice versa. Thus, the satisfying assignments of $\phi$ and the configurations of the Ising model are in one-to-one correspondence.

\begin{table}
	\caption{Truth-table of formulae $p_{ij}=[(x_i \Leftrightarrow x_j) \Leftrightarrow x_{ij}]$ and sign of $\sigma_i\sigma_j$. Note that every assignment $\tau$ such that $p_{ij}(\tau)=1$ is consistent with encoding sough after in Eq.~(\ref{eq:sigma2tau}).}
	\label{tab:predicate}
	\begin{tabular*}{\columnwidth}{@{\extracolsep{\fill}} c c c  c  c c c c}
	\hline
	\hline
	$x_i$ & $x_j$ & $x_{ij}$  & $p_{ij}$ &	$\sigma_i$ & $\sigma_j$ & $\sigma_i \sigma_j$ & Eq.~(\ref{eq:sigma2tau}) \\
	\hline
	0 & 0 & 1 & 1 & $-$ & $-$ & $+$ & True \\
	1 & 0 & 0 & 1 & + & $-$ & $-$ &  True\\
	0 & 1 & 0 & 1 &  $-$ & $+$ & $-$ & True\\
	1 & 1 & 1 & 1 & $+$ & $+$ & $+$ & True\\
	\hline 
	0 & 0 & 0 & 0 & $-$ & $-$ & $+$ & False \\
	1 & 0 & 1 & 0 & + & $-$ & $-$ & False \\
	0 & 1 & 1 & 0 & $-$ & $+$ & $-$ & False \\
	1 & 1 & 0 & 0 & $+$ & $+$ & $+$ & False\\
	\hline
	\hline
	\end{tabular*}
\end{table}

\subsubsection{Step 2 of 2: The literal-weight function}
To construct the literal-weight function $W$ that assigns weights to variables depending on the values they take, first note that the Boltzmann weight of an arbitrary configuration $\sigma$ of the Ising model can be written as
\begin{equation}\label{eq:boltzmann}
	e^{-\beta H(\sigma)} = \left(\prod_{i\in \Lambda} e^{\beta\mu h_i \sigma_i }\right) \cdot\left( \prod _{i,j\in \Lambda} e^{\beta J_{ij} \sigma_i \sigma_j }\right).
\end{equation}
We construct $W$ such that the first term in the right-hand-side of Eq.~(\ref{eq:boltzmann}) is the product of literal weights $W(x_i,\tau(x_i))$, for every $i\in\Lambda$, and the second term of the same expression is the product of literal weights $W(x_{ij},\tau(x_{ij}))$, for every $i\neq j$. Specifically, we introduce the literal-weight function:
\begin{equation}\label{eq:lwf}
    \begin{aligned}
        &W(x_{i},\tau(x_{i})) = 
		e^{\beta \mu h_i (2\tau(x_i) - 1)}, \mbox{ for } i\in \Lambda,\\
		&W(x_{ij},\tau(x_{ij})) = e^{\beta J_{ij} (2\tau(x_{ij}) - 1)}, \mbox{ for } i\neq j,
    \end{aligned}
\end{equation}
and observe that when $\phi(\tau)=1$, or equivalently when Eq.~(\ref{eq:sigma2tau}) maps $\tau$ to an Ising model configuration $\sigma$, we find that the product of literal weights, $\prod_{x\in X}W(x,\tau(x))$, is equal to the Boltzmann weight $e^{-\beta H(\sigma)}$ of the respective Ising model configuration.

\subsubsection{The partition function equals the weighted model count}
Given the Boolean formula $\phi$ in Eq~(\ref{eq:phi}), and given the literal weight function $W$ in Eq.~(\ref{eq:lwf}), we can verify the equivalence between the partition-function value $Z_{\beta}$ and the weighted model count, $W(\phi)$, by writing the latter as a summation of the form $\sum_{\tau \in R} \prod_{x \in X}W(x,\tau(x))$, where $R=\{\tau:\phi(\tau)=1\}$ is the set of satisfying assignments of $\phi$. Moreover, Eq.~(\ref{eq:sigma2tau}) establishes a one-to-one correspondence between the satisfying assignments $\tau \in R$ and the configurations $\sigma\in\{\pm1\}^{\Lambda}$ of the Ising model. Since $\prod_{x \in X}W(x,\tau(x)) = e^{-\beta H(\sigma)}$ for every satisfying assignment $\tau\in R$, we conclude that $Z_{\beta}=W(\phi)$.

Next, we review off-the-shelf weighted model counters and evaluate their performance relative to state-of-the-art solvers (developed by the physics community)
for computing Ising partition function. 

\subsection{Weighted model counters}
Increasingly, model counters are finding applications in fundamental problems of science and engineering, including computations of the normalizing constant of graphical models~\cite{chavira2007compiling}, the binding affinity between molecules~\cite{Viricel2016}, and the reliability of networks~\cite{duenas2017counting}. Thus, the development of model counters remains an active area of research that has seen significant progress over the past two decades~\cite{birnbaum1999good}.

In terms of their algorithmic approaches, exact weighted model counters can be grouped into three broad categories: direct reasoning, knowledge compilation, and dynamic programming. Solvers using direct reasoning (e.g., Cachet~\cite{sang2004combining}) exploit the structure of the CNF representation of the input formula to speedup computation. Solvers using knowledge compilation (e.g., miniC2D~\cite{OD15}) devote effort to converting the input formula to an alternate representation in which the task of counting is computationally efficient, thereby shifting complexity from counting to compilation. Solvers using dynamic programming (e.g., DPMC~\cite{dudek2020dpmc} and TensorOrder~\cite{dudek2020parallel}) exploit the clause structure of the input formula to break the counting task down to simpler computations; e.g., by using graph-decomposition algorithms in the constraint graph representation of the formula. 

In addition to exact solvers, there is a pool of approximate model counters that seeks to rigorously trade off the speed of approximate computations by an admissible error and a level of confidence specified by the user~\cite{CMV21}. The numerical experiments in the next section use representative solvers from the three groups of exact solvers we outlined earlier, as well as an approximate solver.

\subsection{Numerical experiments}
We empirically demonstrate the utility of weighted model counters in Ising model partition function computations by comparisons with the state-of-the-art approximate tool of Ref.~\cite{pan2020contracting}. There, it was shown that tensor network contraction and a matrix product state calculus outperform other strategies in computational physics to obtain machine-precision approximations of the free energy, denoted as $F$. The free energy and the partition function are related as $F=-(1/\beta)\ln Z_{\beta}$. Hereafter, we refer to the publicly available implementation of their method as CATN~\footnote{\url{https://github.com/panzhang83/catn}}. The latter can outperform various mean-field methods with a small runtime overhead, so we take CATN as a baseline for comparison in our experiments. It is worth highlighting that, unlike CATN, model counters are exact or have accuracy guarantees.

Our computational evaluation includes random regular graphs of degree three and two-dimensional $L\times L$ lattices, where $L$ is referred to as the linear size. In particular, square lattices are a standard type of model used to benchmark performance in the Ising context~\cite{pan2020contracting}. In addition, we consider random graphs because their disordered structure provides opportunities to challenge mature WMC solvers as well as physics solvers tailored to the problem.
Each experiment was run in a high performance cluster using a single 2.60 GHz core of an Intel Xeon CPU E5-2650 and 32 GB RAM. Each implementation was run once on each benchmark. An experiment is deemed to be successful when there are no insufficient memory errors and the runtime is within 1000 seconds. We provide all code, benchmarks, and detailed data from benchmark runs at {\Repository}.

The results of the successful experiments are summarized in Fig.~\ref{fig:benchmark}. Significantly, for every instance, there is an exact weighted model counter that outperforms the state-of-the-art approximate tool (CATN). In particular, TensorOrder was faster than CATN in all instances. We stress that partition function computations by model counters are exact or have accuracy guarantees. Overall, for these instances, the dynamic-programming solvers (DPMC and TensorOrder) displayed the best performance, especially for instances generated via random regular graphs. Furthermore, the ratio between CATN and TensorOrder's runtimes grows with problem size, suggesting a substantial algorithmic improvement rather than better hardware utilization or more efficient implementation. We attribute this to the maturity of dynamic-programming model counters, which not only excel at finding optimal solution strategies in structured instances, like the square lattice, but also display excellent performance for randomly generated topologies. The latter case is typically more challenging for less mature heuristics. Also, miniC2D (knowledge compilation) seems to be better suited than Cachet (direct reasoning). 

ApproxMC, the approximate model counter used, fails on all benchmarks except those of extremely small size because it requires that weighted model counting problems be encoded as unweighted problems, greatly increasing the number of variables and clauses in the formulae it operates over. However, ApproxMC will benefit from future encodings, and still offer its desirable guarantees on the error and confidence of approximations. 

Overall, in our experiments, CATN and TensorOrder were the only solvers that did not timeout. In addition, TensorOrder's exact calculations were generally faster than CATN’s approximate computations, at times an order of magnitude faster. We note that both tools are tensor network contraction-based; however, TensorOrder is a dynamic programming solver that uses state-of-the-art graph decomposition techniques to select contraction orderings. We take the better performance of TensorOrder, especially in random regular graphs, as evidence that its contraction orderings are more robust than those used by CATN.

The inverse temperature parameter in our numerical experiments ranges from 0.1 to 2.0, as we consider the same Ising model instances in the work of Pan et al.~\cite{pan2020contracting}. Typically, weighted counting solvers use fixed-point arithmetic that can lead to numerical errors due to decimal truncation. We did not detect significant numerical errors for the Ising model instances in this study; however, we expect numerical errors to be significant for large inverse temperature values of $\beta$. In the experimental settings considered, Solvers' runtimes were not found to vary with $\beta$.

We believe that the computational advantage of weighted counters over state-of-the-art tools in this study is partly due to their ability to take advantage of treelike structures of Ising model instances. Thus, the computational advantage of weighted counters will hold in graphs where the treewidth stays fixed as the number of spins increases. In contrast, for graphs where the treeewith increases as the number of spins increases, like the square lattice, weighted counters will face the same challenges of every other tool, exact or approximate.

In sum, our empirical evaluation shows that weighted model counters can vastly outperform specialized tools in the physics community in computations of the partition functions of the Ising model. Additionally, the exact nature of computations makes model counters more robust to alternatives such as CATN, which is an approximate tool that lacks guarantees of accuracy.
\begin{figure*}[hbt!]
	\centering
	\includegraphics[]{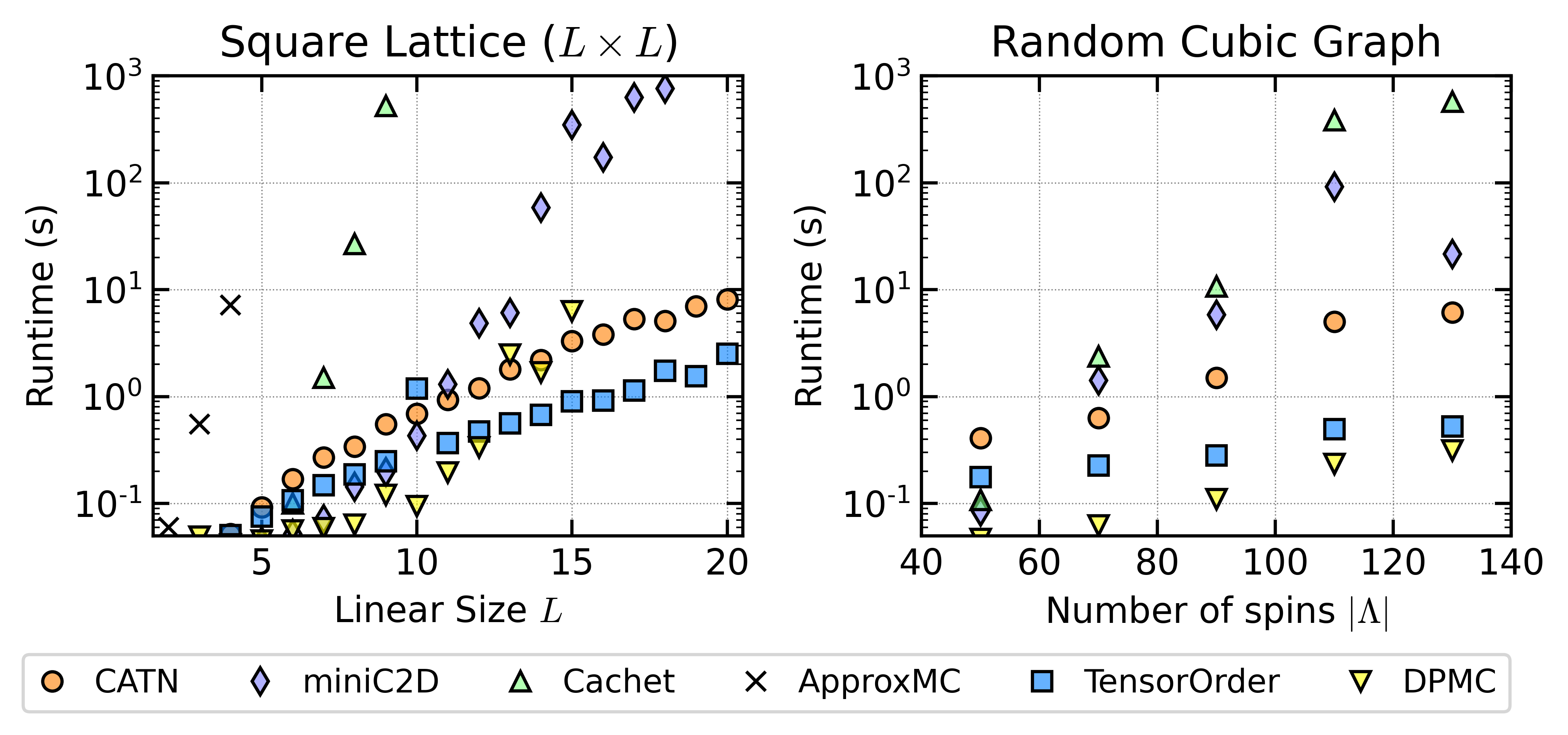}
	\caption{Running time of the partition function computation for the exact weighted model counters, miniC2D~\cite{OD15}, Cachet~\cite{sang2004combining}, TensorOrder~\cite{dudek2020parallel} and DPMC~\cite{dudek2020dpmc}; the approximate model counter, ApproxMC~\cite{CMV16-AppMC}; and the approximate reference tool from physics, CATN~\cite{pan2020contracting}. Note in particular that the ratio between CATN and TensorOrder's runtimes grows with problem size. This suggests that the difference in performance is due to a substantial algorithmic improvement rather than better hardware utilization or more efficient implementation. Left: Two-dimensional square lattice with linear size $L$. Right: Random regular graphs with degree three.}
	\label{fig:benchmark}
\end{figure*}

%% file: sections/5-Conclusions.tex
\section{Conclusions}
\label{sec:conc}

The principal aim of our work has been to demonstrate the utility of casting Ising-model partition-function computation to standard forms in the field of computer science (i.e., \#wCSP and WMC). Because these standard forms are well-studied, they provide immediate theoretical and computational capabilities to those studying the problem without the need for laborious, handcrafted techniques as in \cite{IsingSingleHardness} and \cite{pan2020contracting}. 

By representing partition-function computation as a w\#CSP, we were able to determine the computational complexity of computing Ising model partition functions with minimal effort. While lacking the flexibility of bespoke reductions, the w\#CSP framework serves as a very easy, out of the box way to determine the hardness of some problems for Ising models. This w\#CSP representation also provided some intuition on where the difficulty of Ising partition function computation comes from (i.e., the difficulty in relating the change of a lattice site's spin to a change in probability across configurations).

From our \#wCSP formulation, we were inspired to develop a reduction from Ising-model partition-function computation to weighted model counting (WMC), a well-studied problem for which many mature, off-the-shelf solvers exist. Some exact WMC solvers were able to out perform state-of-the-art approximate tools for computing partition functions, running in a tenth of the time despite the state-of-the-art tools being specially designed for Ising models. This can be largely attributed to the WMC solvers' maturity. Furthermore, as weighted model counters continue to improve, Ising-model partition-function computation performed using these solvers will improve in lockstep.

While we have focused primarily on the Ising model in this paper, the same type of analysis is easily done for Potts models and other related models. Our hope is that by demonstrating the utility of the w\#CSP and WMC frameworks in this setting, readers will be encouraged to apply them to other settings and problems of interest in statistical physics or probabilistic inference. 

Furthermore, we believe that weighted counters could be applied in the computation of the partition function of non classical models such as the quantum Ising model. The main justification is that computing the ground state degeneracy and density of states for both classical and quantum Hamiltonians is computationally equivalent~\cite{Brown2011}. Thus, a systematic study of the computational complexity of the partition function of non classical models and the promise of using existing weighted counters in this context is left as a promising avenue for future research.


\section{Acknowledgements}
This work was supported in part by the Big-Data Private-Cloud Research Cyberinfrastructure MRI-award funded by NSF under grant CNS-1338099, a Civil Infrastructure Systems award also funded by NSF through grant CMMI-2037545, and by Rice University's Center for Research Computing (CRC).

The authors are also grateful for discussions with members of the Quantum Algorithms Group at Rice University, including professors Kaden Hazzard and Guido Pagano (Physics and Astronomy), as well as Anastasios Kyrillidis (Computer Science).  

%% file: sections/AppendixA-WCSPtoWMCReduction.tex
\section{Reducing w\#CSPs to Weighted Model Counting}
\label{appendix:wcsp2wmc}

Every instance of a weighted counting constraint-satisfaction problem (w\#CSP, Def.~\ref{restatdef:wcspInstanceDef}) can be reduced to an equivalent instance of weighted model counting (WMC, Def.~\ref{def:wmc}) expressed in conjunctive normal form (CNF, Def.~\ref{def:cnf}). One particular reduction for computing the partition function of the Ising model is given in Section~\ref{sec:wmc}. In this appendix, we give a more general reduction from an aribtrary w\#CSP to WMC:

\begin{lemma}
For every constraint language $\F$, there is a polynomial-time reduction reduction from $\#CSP(\F)$ to WMC.
\end{lemma}

\begin{proof}
Consider a constraint language $\F$ with a domain $D= \{0, 1, \cdots, \abs{D}-1\}$, and the associated problem $\#CSP(\F)$. Consider an instance $(I, n)$ over the variables $x_1, \cdots, x_n$. Let $\overline{D} = \lceil \log_{2}(\abs{D}) \rceil$. We reduce $(I,n)$ to a WMC instance as follows:

We first define the set $X$ of Boolean variables in the WMC instance. 
\begin{itemize}
    \item $x_{j,1}, \cdots, x_{j,\overline{D}}$ - For each of the $n$ variables $x_j$ in the input instance, we add $\overline{D}$ corresponding Boolean variables $x_{j,1}, \cdots, x_{j,\overline{D}}$ to $X$. These variables encode the value of each $x_j$. In particular, the value of $x_j$ corresponds to the binary number that results from concatenating the values of Boolean variables: $x_{j,\overline{D}} \cdots x_{j,1}$.
    \item $x_{(i_1,\cdots,i_k), \boldsymbol{d}}$ - For each formula $F(x_{i_1},\cdots,x_{i_k})$ in $I$, we add $\abs{D}^k$ Boolean variables $x_{(i_1,\cdots,i_k), \boldsymbol{d}}$, one for each $\boldsymbol{d} \in D^k$. These variables encode 
    the value of the input variables $(x_{i_1},\cdots,x_{i_k})$ of each formula $F(x_{i_1},\cdots,x_{i_k})$ in $I$.
\end{itemize}
We now introduce the following formulas:
\begin{itemize}
    \item First, we must restrict the allowed values of the variables $x_{j,1}, \cdots, x_{j,\overline{D}}$ so that the value of $x_j$ is in $D$.  Specifically, we want to disallow the binary number $x_{j,\overline{D}} \cdots x_{j,2}x_{j,1}$ from exceeding $\abs{D}-1$. The most straightforward way to express this condition is that $x_{j,\overline{D}} \cdots x_{j,2}x_{j,1}$ should not be greater than the binary representation of $\abs{D}-1$, written $B_{\overline{D}} \cdots B_2B_1$, with respect to the following strict lexicographic order ($<$) on binary strings of equal length:
    \begin{itemize}
        \item[$\blacksquare$] $0 < 1$
        \item[$\blacksquare$] $0\omega_a < 1\omega_b$ for arbitrary binary strings $\omega_a$ and $\omega_b$ of equal length.
        \item[$\blacksquare$] $0\omega_a < 0\omega_b$ and $1\omega_a < 1\omega_b$ if $\omega_a < \omega_b$ for arbitrary binary strings $\omega_a$ and $\omega_b$ of equal length.
    \end{itemize}
    For $a,b \in \{0,1\}$ and binary strings of equal length $\omega_a$ and $\omega_b$, we have that $a\omega_a < b\omega_b$ if $a < b$ or if $a = b$ and $\omega_a < \omega_b$. We now construct inductively a formula $\phi(a_m \cdots a_1, b_m \cdots b_1)$ that checks whether $a_m \cdots a_1 < b_m \cdots b_1$, defined as follows: 
    \begin{align}
        \phi&(a_m \cdots a_1, b_m \cdots b_1)\\
        &= \begin{cases}
        (a_m < b_m) & m > 1 \\\quad \lor \left((a_m = b_m) \land \phi(a_{m-1} \cdots a_1, b_{m-1} \cdots b_1)\right) \\
        (a_m < b_m)  & m = 1\\
        \end{cases}\\
        &= (a_m < b_m)\\
            &\quad \lor ((a_m = b_m) \land (a_{m-1} < b_{m-1})\nonumber\\
            &\quad \lor ((a_m = b_m) \land (a_{m-1} = b_{m-1}) \land (a_{m-2} < b_{m-2}))\nonumber\\
            &\quad \lor \cdots\nonumber\\
        &= \bigvee\limits_{p \in [m]} \left(\left(\bigwedge\limits_{q \in [m], q > p} (a_q = b_q)\right) \land (a_p < b_p) \right)\\
        &= \bigvee\limits_{p \in [m], a_p = 0} \left(\left(\bigwedge\limits_{q \in [m], q > p} (a_q = b_q)\right) \land (a_p < b_p) \right)\displaybreak\\
        &= \bigvee\limits_{p \in [m], a_p = 0} \left(\left(\bigwedge\limits_{q \in [m], q > p} (a_q = b_q)\right) \land (b_p = 1) \right).
    \end{align}
    The last two simplifications come from the fact that $a_p < b_p$ can only be true if $a_p = 0$ and $b_p = 1$.
    
    Next, we construct a Boolean formula to capture $\neg \phi(B_{\overline{D}} \cdots B_2B_1, x_{j,\overline{D}} \cdots x_{j,2}x_{j,1})$. 
    
    Now define
        \begin{equation}
            \phi_j = \neg \phi\left(B_{\overline{D}} \cdots B_2B_1, x_{j,\overline{D}} \cdots x_{j,2}x_{j,1}\right),
        \end{equation}
    and note that terms in $\phi_j$ checking equality can be replaced with positive and negative literals over $x_{j,\overline{D}}, \cdots, x_{j,2}, x_{j,1}$ determined by the values of $B_{\overline{D}}, \cdots, B_2, B_1$. As the negation of a DNF formula, $\phi_j$ is a CNF formula.
    
    We now have a formula $\phi_j$ that is satisfied exactly when $x_{j,\overline{D}} \cdots x_{j,2}x_{j,1} \leq \abs{D}-1$.
    \item Next, we want to relate the variables that encode the inputs to the constraint functions (i.e. $x_{(i_1,\cdots,i_k), \boldsymbol{d}}$) to the variables encoding the value of $x_j$ (i.e. $x_{j,1}, \cdots, x_{j,\overline{D}}$). To do this, we introduce the following shorthand:
    \begin{align}
        \phi_{j,d} = \bigwedge\limits_{q \in [\overline{D}]} l_{j,q,d}
    \end{align}
    where $d \in D$ with the binary representation $d_{\overline{D}} \cdots d_1$, and $l_{j,q,d} = \begin{cases} x_{j,q} & d_q = 1\\ \neg x_{j,q} & d_q = 0\end{cases}$.
    
    Evidently, $\phi_{j,d}$ is true exactly when $x_j = d$. We also introduce the following formulae:
        \begin{align}
        \psi_{F(x_{i_1},\cdots,x_{i_k})} &= \bigwedge\limits_{\boldsymbol{d} \in D^k} \left(x_{(i_1,\cdots,i_k), \boldsymbol{d}} \implies \bigwedge\limits_{p \in [k]} \phi_{i_p,\boldsymbol{d}_p}\right)\\
        &= \bigwedge\limits_{\boldsymbol{d} \in D^k} \left(\bigwedge\limits_{p \in [k]} \left(\neg x_{(i_1,\cdots,i_k), \boldsymbol{d}} \lor \phi_{i_p,\boldsymbol{d}_p}\right)\right)\\
        &= \bigwedge\limits_{\boldsymbol{d} \in D^k,p \in [k]} \left(\neg x_{(i_1,\cdots,i_k), \boldsymbol{d}} \lor \bigwedge\limits_{q \in [\overline{D}]} l_{j,q,\boldsymbol{d}_p}\right)\displaybreak\\
        &= \bigwedge\limits_{\boldsymbol{d} \in D^k,p \in [k]} \left(\bigwedge\limits_{q \in [\overline{D}]} \left(\neg x_{(i_1,\cdots,i_k), \boldsymbol{d}} \lor l_{j,q,\boldsymbol{d}_p}\right)\right)\\
        &= \bigwedge\limits_{\boldsymbol{d} \in D^k,p \in [k], q \in [\overline{D}]} \left(\neg x_{(i_1,\cdots,i_k), \boldsymbol{d}} \lor l_{j,q,\boldsymbol{d}_p}\right)\\
        &\mbox{ and}\nonumber 
        \\\gamma_{F(x_{i_1},\cdots,x_{i_k})} &= \bigvee\limits_{\boldsymbol{d} \in D^k} x_{(i_1,\cdots,i_k), \boldsymbol{d}}
    \end{align}
    which together force $x_{(i_1,\cdots,i_k), \boldsymbol{d}}$ to be true exactly when the value assigned to $(x_{i_1},\cdots,x_{i_k})$ is $\boldsymbol{d}$. Specifically, $\psi_{F(x_{i_1},\cdots,x_{i_k})}$ ensures that if $x_{(i_1,\cdots,i_k), \boldsymbol{d}}$ is true then the value assigned to $(x_{i_1},\cdots,x_{i_k})$ is $\boldsymbol{d}$. Since the value assigned $(x_{i_1},\cdots,x_{i_k})$ will be equal to exactly one $\boldsymbol{d} \in D^k$, $x_{(i_1,\cdots,i_k), \boldsymbol{d}}$ can be true for at most one $\boldsymbol{d} \in D^k$. Moreover, from $\gamma_{F(x_{i_1},\cdots,x_{i_k})}$, we have that $x_{(i_1,\cdots,i_k), \boldsymbol{d}}$ is true for at least one $\boldsymbol{d} \in D^k$. Thus, $x_{(i_1,\cdots,i_k), \boldsymbol{d}}$ will be true exactly when the value assigned to $(x_{i_1},\cdots,x_{i_k})$ is $\boldsymbol{d}$.
\end{itemize}
All together, we write 
\begin{align}
    \Phi &= \left(\bigwedge\limits_{j \in [n]} \phi_{j}\right)\nonumber\\ &\land \bigwedge\limits_{F(x_{i_1},\cdots,x_{i_k}) \in I} \left(\psi_{F(x_{i_1},\cdots,x_{i_k})} \land \gamma_{F(x_{i_1},\cdots,x_{i_k})}\right).
\end{align}
This formula $\Phi$ is the formula whose weighted count our WMC reduction of the $(I,n)$ instance will compute.

Finally, we need a weight function $W$. For each $x_{j,q}$, $W(x_{j,q}, 1) = W(x_{j,q}, 0) = 1$. For each $x_{(i_1,\cdots,i_k), \boldsymbol{d}}$, $W(x_{(i_1,\cdots,i_k), \boldsymbol{d}}, 0) = 1$ and $W(x_{(i_1,\cdots,i_k), \boldsymbol{d}}, 1) = F(\boldsymbol{d})$.

All together, $X$, $\Phi$, and $W$ constitute the WMC instance produced by the reduction.
Correctness of the reduction follows by construction. The variables $x_{j,\overline{D}}, \cdots, x_{j,1}$ give a binary encoding of the value of $x_j$, thanks to $\phi_j$. For each formula $F(x_{i_1},\cdots,x_{i_k}) \in I$, the literals $x_{(i_1,\cdots,i_k), \boldsymbol{d}}$ indicate the value of the arguments of $F$, and the weight assigned to each literal $x_{(i_1,\cdots,i_k), \boldsymbol{d}}$ is exactly the value of $F(\boldsymbol{d})$. The total weight of an assignment to the variables of $\Phi$ is a product of these weights, which is precisely the product of the values taken by the formulas in $I$ for the corresponding assignment to $x_1, \cdots, x_n$. So the weight $W$ on each assignment over the Boolean variables is equal to the value of $F_I$ on the corresponding assignment to $x_1, \cdots, x_n$. Since there is a bijection between assignments to the Boolean variables satisfying $\Phi$ and assignments to $x_1, \cdots, x_n$, we can show $W(\Phi) = Z(I)$. Thus the WMC instance that the reduction produces is equivalent to the $\#CSP(\F)$ instance $(I,n)$ as desired.
\end{proof}

Note that we have $\overline{D} * n$ variables of the form $x_{j, q}$ and at most $\abs{D}^K * \abs{I}$ variables of the form $x_{(i_1,\cdots,i_k), \boldsymbol{d}}$, where $K$ is the maximum arity of the constraints used in $I$. There are $n$ formulae of the form $\phi_j$ in $\Phi$, each with $\overline{D}$ literals. There are $\abs{I}$ formulae of the form $\psi_{F(x_{i_1},\cdots,x_{i_k})}$ in $\Phi$, each with at most $\abs{D}^K * K*\overline{D}$ CNF clauses, each clause having $2$ literals. Finally, there are $\abs{I}$ formulas of the form $\gamma_{F(x_{i_1},\cdots,x_{i_k})}$, each with at most $\abs{D}^K$ literals. The reduction of $(I,n)$ is exponential only in $K$, the maximum arity of the functions appearing in $I$. 

If the arities of the constraints comprising $\F$ are bounded, then the reduction is polynomial. Any finite constraint language $\F$ has bounded arity, but many useful infinite constraint languages have bounded arity as well. In the case of $\#CSP(\text{Ising})$, the functions in the associated (infinite) constraint language have arity at most $2$.

For many particular problems, a more compact reduction can be found. For example, when there is symmetry in the constraints $F$ in $\F$ (i.e., when the constraints are not injective), more efficient encodings of the constraints' inputs can be used.

%% file: sections/AppendixB-DichotomyThm.tex
\section{Hardness and Relationship to Weighted Constraint Satisfaction}
\label{appendix:dichot}

Recall the discussion from Section~\ref{sec:theory:dichot}. We consider there a problem instance $(I,n)$ and the associated instance function $F_I$. For each $t < n$, we partition the partial assignments $D^{t-1}$ into equivalence classes $S_{[t,j]}$ based on the contributions of each partial assignment to $Z(I)$. We take $m_t$ to be the number of these equivalence classes for each $t$. The discussion then proceeds with an analysis of these objects to determine conditions under which a problem is polynomially solvable. Section~\ref{sec:theory:dichot} introduces the Block-Orthogonality Condition. Here we review the remaining two conditions: the Mal'tsev Condition and the Type-Partition Condition.

The Block-Orthogonality Condition given in Section~\ref{sec:theory:dichot} gives us that $m_t$ is reasonably small. This means that the speedup we get from using the equivalence relation will be substantial, provided we are able to easily identify equivalent partial assignments.

To make this identification, we require that membership in $S_{[t,j]}$ is computable in polynomial time. This requirement is guaranteed by the \emph{Mal'tsev Condition}. In particular, the Mal'tsev Condition, by requiring the existence of Mal'tsev polymorphisms (defined below) for each $S_{[t,j]}$, guarantees the existence of a witness function for each $S_{[t,j]}$ whose evaluation time is linear in $t$. For our purposes, a witness function for a set is a function that verifies whether a given input is an element of that set. Further details about such witness functions and their construction is given in~\cite{CaiComplex}.

Understanding the Mal'tsev Condition requires us to define polymorphisms and Mal'tsev polymorphisms. 

\begin{definition} [(Cubic) Polymorphism]
A cubic polymorphism (polymorphism, for short) of $\Phi \subset D^t$ is a function $\phi:D^3 \rightarrow D$ such that, for all $\boldsymbol{u}, \boldsymbol{v}, \boldsymbol{w} \in \Phi$, $(\phi({u}_1, {v}_1, {w}_1), \cdots, \phi({u}_t, {v}_t, {w}_t)) \in \Phi$.
\end{definition} 
\begin{definition} [Mal'tsev Polymorphism]
A Mal'tsev polymorphism of $\Phi \subset D^t$ is a polymorphism $\phi:D^3 \rightarrow D$ of $\Phi$ such that, for all $a,b \in D$, we have that $\phi(a,a,b) = \phi(b,a,a) = a$.
\end{definition} 
\begin{definition} [Mal'tsev Condition] \label{condition:maltsev} We say that a constraint language $\F$ satisfies the Malt'sev Condition if, for every instance function $F_I$ of $\#CSP(\F)$, 
all the equivalence classes $S_{[t,j]}$ associated with $F_I$ share a common Mal'tsev polymorphism.
\end{definition}
    
With the Mal'tsev Condition, we know that there exists a tractable witness function for each $S_{[t,j]}$. Given sufficient information about each $S_{[t,j]}$, we might hope to compute such witness functions. However, we will in general not know $m_t$, nor will we know much about each $S_{[t,j]}$. The existence of a \textit{shared} Mal'tsev polymorphism, along with the \emph{Type-Partition Condition} below, allows us to overcome this lack of information. When these conditions are satisfied, we can determine each $m_t$ and construct witness functions for each $S_{[t,j]}$ without requiring prior knowledge or construction of each $S_{[t,j]}$.

We present the Type-Partition Condition below. A complete coverage can be found in~\cite{CaiComplex}.

\begin{definition} [Prefix]
Let $n,m \in \Z_+$ with $n \geq m$ be given. Let $\boldsymbol{v} = (v_1, \cdots, v_n) \in \R^n$ be given. We say a vector $\boldsymbol{w} \in \R^m$ is a prefix of $\boldsymbol{v}$ if, for all $j \leq m$, $w_j = v_j$. Thus $\boldsymbol{v} = (w_1, \cdots, w_m, v_{m+1}, \cdots, v_n)$, so $\boldsymbol{w}$ appears as a prefix of $\boldsymbol{v}$.
\end{definition}

\begin{definition} [$type_F$]
Let a function $F: D^n \rightarrow \C$ with $n \geq 2$ be given. Let $P([m])$ denote the power set of $\{1, \cdots, m\}$. We define the map $type_F: \bigcup\limits_{k \in [n]} D^k \rightarrow P([m])$ so that for each $\boldsymbol{y} \in \bigcup\limits_{k \in [n]} D^k$, we have $type_F(\boldsymbol{y}) = \{j \in [m]: \boldsymbol{y} \text{ is a prefix of an element of } S_{[n,j]}\}$.
\end{definition}

\begin{definition} [Type-Partition Condition] \label{condition:type-partition}
We say that a constraint language $\F$ satisfies the Type-Partition Condition if, for all instance functions 
$F_I$ in $\#CSP(\F)$ with arity $n \geq 2$, for all $t \in [n]$, $l \in [t-1]$, and $\boldsymbol{y}, \boldsymbol{z} \in D^l$, we have that the sets $type_{F^{[t]}_I}(\boldsymbol{y})$ and $type_{F^{[t]}_I}(\boldsymbol{z})$ are either equal or disjoint. 
\end{definition}

All together, the criteria for the dichotomy theorem give us the ability to make real use of our equivalence relation $\equiv_t$. The Block-Orthogonality Condition guarantees that we do not have too many equivalence classes. The Mal'tsev and Type-Partition Conditions together let us construct witness functions that we use to identify which equivalence class a given partial assignment belongs to. With a guarantee that our equivalence relation is useful and the ability to determine elements' equivalence classes, we are able to use the methodology described in Section~\ref{sec:theory:dichot} to determine $Z(I)$ in polynomial time.

This provides some intuition that the dichotomy theorem criteria are sufficient for a problem to be in $FP$. For a proof that w\#CSPs that do not obey the dichotomy theorem criteria are \#P-hard, refer to~\cite{CaiComplex}.